\documentclass[11pt]{article}

\usepackage{amsmath,amssymb,amsthm, latexsym,color,epsfig,enumerate,a4, graphicx}
\usepackage{algorithmic}
\usepackage{algorithm}
\usepackage{tikz}
\usetikzlibrary{decorations.pathreplacing}

\theoremstyle{plain}
\newtheorem{theorem}{Theorem}[section]
\newtheorem{lemma}[theorem]{Lemma}

\newtheorem{proposition}[theorem]{Proposition}

\newtheorem*{theorem*}{Theorem}
\newtheorem*{claim*}{Claim}
\newtheorem*{claims*}{Claims}
\newtheorem*{lemma*}{Lemma}
\newtheorem*{corollary*}{Corollary}
\newtheorem*{definition*}{Definition}

\newcommand{\eps}{\varepsilon}

\date{}
\title{\vspace{-0.7cm} Unique reconstruction threshold for random jigsaw puzzles}
\author{
Rajko Nenadov \thanks{Institute of
Theoretical Computer Science, ETH Z\"urich, 8092 Z\"urich,
Switzerland. Emails: {\tt rnenadov@inf.ethz.ch}, {\tt ppfister@inf.ethz.ch}, {\tt steger@inf.ethz.ch}. } 
\and Pascal Pfister $^*$
\and Angelika Steger$^*$ 
 }

\begin{document}
\maketitle

\begin{abstract}
A random jigsaw puzzle is constructed by arranging $n^2$ square pieces into an $n \times n$ grid and assigning to each edge of a piece one of $q$ available colours uniformly at random, with the restriction that touching edges receive the same colour. We show that if $q = o(n)$ then with high probability such a puzzle does not have a unique solution, while if $q \ge n^{1 + \varepsilon}$ for any constant $\varepsilon > 0$ then the solution is unique. This solves a conjecture of Mossel and Ross ({\em Shotgun assembly of labeled graphs}, arXiv:1504.07682).
\end{abstract}

\section{Introduction}
In this paper we study the \emph{random jigsaw puzzle} problem posed by Mossel and Ross \cite{mossel2015shotgun}. Suppose we are given a square wooden board 1cm thick. One way to make a puzzle out of such a board is as follows:
\begin{enumerate}[1.]
	\item Cut the board into $n^2$ smaller squares of the same size. We call these squares \emph{pieces}.
	\item Arrange the pieces into an $n \times n$ grid and draw a unique, non-symmetrical piece of art on each of them. The resulting big picture is referred to as the \emph{puzzle picture}.
	\item Colour each side (other than the ones facing up and down) of a piece  with one of the $q$ available colours such that the touching sides of two adjacent pieces receive the same colour. 
\end{enumerate}
Suppose now that we shuffle all the pieces and give them to a friend. This friend would like to find out what the puzzle picture is, and the only way she can do it is by arranging the pieces back into an $n \times n$ grid such that the touching sides of each two adjacent pieces have the same colour (we call such an arrangement \emph{valid}). Here we allow each piece to be arbitrarily rotated, however, once arranged in the grid, each piece should be placed such that the drawing on it is visible (i.e. flips are not allowed). This is slightly more general than the model used in \cite{mossel2015shotgun} where no rotations were allowed.

Note that each puzzle created in the above manner contains a valid arrangement, namely the original construction. However, there might be other valid arrangements which, due to the uniqueness and non-symmetry of drawings on pieces, necessarily produce a picture different than the puzzle picture. For example, if $q = 1$, then in fact any arrangement is valid. As the puzzle picture might contain a secret message for our friend, we would like to design a puzzle such that there exists a unique valid arrangement. We make this precise in the next section. A simple way to make a puzzle with a unique solution is to have $q = 2n^2 + 2n$ colours available and use each colour exactly once. However, in this case it is computationally trivial to reconstruct the puzzle picture by simply starting with an arbitrary piece and then putting down any valid piece in each next step. Therefore, to make the puzzle interesting, we would like that (i) there exists exactly one valid arrangement and (ii) the set of colours is not `too large' or, alternatively, most of the colours are used a significant number of times.

It turns out that a \emph{random jigsaw puzzle} satisfies these two properties. A random jigsaw puzzle is created by colouring sides of pieces independently and uniformly at random using one of the $q$ colours, respecting that once a side has been coloured the corresponding side of the adjacent piece gets the same colour. Mossel and Ross \cite{mossel2015shotgun} studied the following question: how large does $q$ have to be to ensure that a random jigsaw puzzle has with high probability (w.h.p, meaning with probability $1-o(1)$ as $n\rightarrow \infty$) only one possible arrangement? They showed that the probability of having a unique reconstruction goes to $0$ if $q = o(n^{2/3})$ and that it goes to $1$ if $q = \omega(n^{2})$ (as $n \rightarrow \infty$). Furthermore, they conjectured that there exists a constant $c$ such that for all $\eps  > 0$ the probability of unique reconstruction is $o(1)$ (resp.\  $1-o(1)$) if  $q \le n^{c-\eps}$ (resp.\ $q \ge n^{c+\eps}$). From Mossel~\cite{MosselPC} we  learned that he and Uri Feige  conjecture that $c=1$ is the correct value.
We answer this conjecture in the affirmative. 

\begin{theorem} \label{thm:Jigsaw}
Let $\varepsilon > 0$ be a constant. If $q = o(n)$ (resp. $q \ge n^{1 + \varepsilon}$), then the probability that a random jigsaw puzzle has a unique reconstruction (up to rotation of the whole grid) goes to $0$ (resp. $1$) as $n \rightarrow \infty$.
\end{theorem}

Our proof is elementary, essentially  based only on first moment arguments. We refer the reader to \cite{alon2015probabilistic} for  an introduction to probabilistic tools.
Note that a simple calculation shows that for $q = n^{1 + \varepsilon}$ all colours are used asymptotically the same number of times, thus giving a simple randomized procedure for creating a puzzle satisfying (i) and (ii).

\section{Notation and preliminaries}

%Let $G$ be an $n \times n$ grid of puzzle pieces. Each piece is enclosed by four \emph{half-edges}. Thus, two \emph{adjacent} pieces share a pair of \emph{adjacent} half-edges. We call a pair of adjacent half-edges an \emph{edge}. Observe that half-edges along the border of the grid have no other half-edge adjacent to them, but, for simplicity, we assume that these `single' half-edges form an edge as well.  We denote by $E_G(V)$ the set of edges in the grid $G$. For any two sets $S, T \subset V$ we denote by $E_G(S,T)$ the edges of which one half-edge belongs to a piece of $S$ and the other belongs to a piece of $T$ and we define $E_G(T) := E_G(T,T)$.
%
%Let $Q$ be a set of possible colours with $|Q|=q$. Let $g: E_G \rightarrow Q$ be a function which assigns colours to edges. In sections 3 - 5 we will always assume that the function $g$ is \emph{random} colouring, i.e. it chooses the colours of the edges independently and uniformly at random.
%

We first introduce some notation. A jigsaw puzzle is given by an $n \times n$ grid $G$ of \emph{puzzle pieces}, which we denote by~$V$. We use $v_{i,j} \in V$ to denote the piece corresponding to the location $(i,j)$ in the grid. Each piece is enclosed by four \emph{half-edges} (top, right, bottom and left side of the piece). Two \emph{adjacent} pieces share a pair of \emph{adjacent} half-edges (e.g. the top half-edge of the lower piece and the bottom half-edge of the upper piece), which we call an \emph{edge}. Observe that half-edges along the border of the grid have no other half-edge adjacent to them but, for simplicity, we assume that these `single' half-edges form an edge as well.  We denote by $E_G$ the set of edges in the grid $G$. Given a set $Q$, a {\em colouring} of $G$ is a function  $g: E_G \rightarrow Q$ that assigns colours to edges and, in turn, to half-edges in the natural way (i.e. half-edges corresponding to an edge inherit the colour). Here we consider the case where $g$ is a random function: the colour of each edge is chosen independently and uniformly at random from $Q$.

% We consider the following question: given a random function $g: E_G \rightarrow Q$ (i.e. the colour of each edge is chosen independently and uniformly at random) is there w.h.p (i.e. with probability $1-o(1)$ as $n \rightarrow \infty$) a unique way to reassemble the puzzle. Clearly, the answer to this question depends on the number $q$ of available colours, where we assume that $q=q(n)$ is a function of $n$.

Given an $n \times n$ grid $G$ and a colouring function $g$, one can do two things to reassemble the puzzle in a non-trivial way: reorder the pieces and rotate them. A reordering of the pieces corresponds to a bijection $\phi \colon n \times n \rightarrow n \times n$. A rotation of a piece corresponds to a cyclic shift of its edges. Therefore, we can describe rotations of all pieces with a function $\Pi$ which assigns any piece $v_{i,j} \in V$ a cyclic shift $\pi_{i,j}$. Any two such functions then define another $n \times n$ grid  of pieces $H = G(\phi, \Pi)$, which we call a \emph{reconstruction}. For a given reconstruction $H$, we denote by $E_H$ the set of its edges (recall that an edge is composed of the half-edges of adjacent pieces). We say that $H$ is a \emph{valid} reconstruction if all edges in $E_H$ are monochromatic, that is, every two half-edges corresponding to an edge have the same colour. Using this notation, Theorem \ref{thm:Jigsaw} can be formally stated as follows. 

\begin{theorem*}[Theorem \ref{thm:Jigsaw} restated]
Let $\varepsilon > 0$ be a constant. If $q = o(n)$ (resp. $q \ge n^{1 + \varepsilon}$) and $g:E_G \rightarrow Q$ is a random colouring of an $n \times n$ jigsaw puzzle $G$ with $|Q| = q$ colours, then w.h.p there exists (resp. does not exist) a valid reconstruction $H = G(\phi, \Pi)$ such that $E_H \neq E_G$, up to rotation of the whole grid.
\end{theorem*}

For any two sets $S, T \subset V$ we denote by $E_G(S,T)$ the edges for which one half-edge belongs to a piece of $S$ and the other belongs to a piece of $T$, and we define $E_G(T) := E_G(T,T)$. We call $T$ a connected set if the vertices in $T$ induce a connected subgraph of the grid, where connectivity is defined as in the graph theoretic setting. 
Figure \ref{fig:defs} illustrates this and some more basic definitions  that we use throughout the paper.

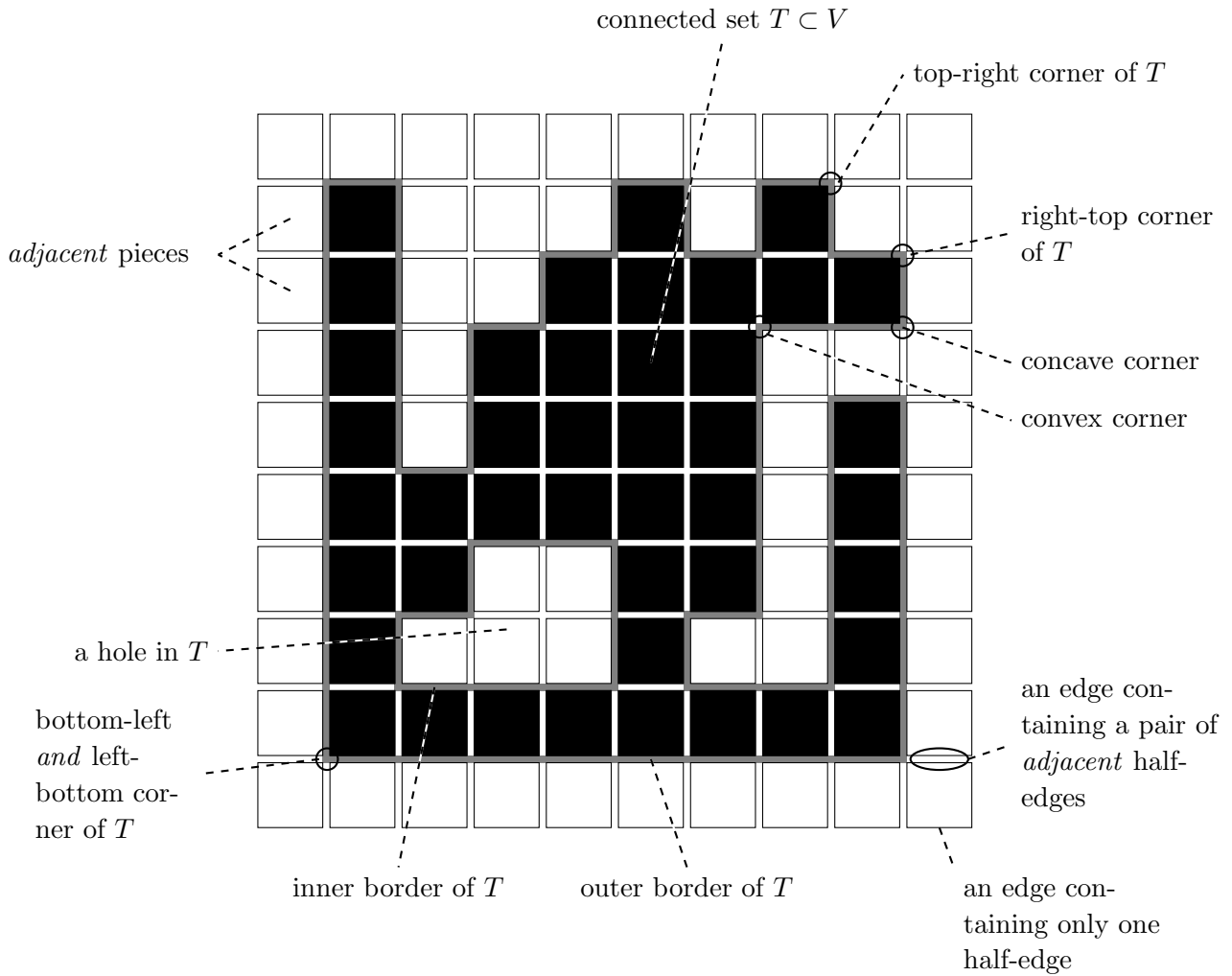
\begin{figure}[h!]
\begin{center}
\begin{tikzpicture}

% hole in T%

\path   (2.05,2.05) rectangle (2.95, 2.95);	
\foreach \x in {3,4}{
	\foreach \y in {2,3}{
		\path  (\x+0.05,\y+0.05) rectangle (\x + 0.95, \y + 0.95);		
	}
}

%inner border in magenta%
\path [fill=gray] (1.95,1.95) rectangle (5.05, 2.05);
\path [fill=gray] (4.95,1.95) rectangle (5.05, 4.05);
\path [fill=gray] (5.05,4.05) rectangle (2.95, 3.95);
\path [fill=gray] (3.05,4.05) rectangle (2.95, 2.95);
\path [fill=gray] (3.05,3.05) rectangle (1.95, 2.95); 
\path [fill=gray] (2.05,3.05) rectangle (1.95, 1.95);

%outer border in red%
\path [fill=gray] (0.95,0.95) rectangle (9.05, 1.05);
\path [fill=gray] (8.95,0.95) rectangle (9.05, 6.05);
\path [fill=gray] (9.05,6.05) rectangle (7.95, 5.95);
\path [fill=gray] (8.05,1.95) rectangle (7.95, 5.95);
\path [fill=gray] (8.05,1.95) rectangle (5.95, 2.05); 
\path [fill=gray] (5.95,1.95) rectangle (6.05, 3.05);
\path [fill=gray] (5.95,2.95) rectangle (7.05, 3.05);  
\path [fill=gray] (6.95,2.95) rectangle (7.05, 7.05); 
\path [fill=gray] (6.95,6.95) rectangle (9.05, 7.05); 
\path [fill=gray] (8.95,6.95) rectangle (9.05, 8.05); 
\path [fill=gray] (9.05,8.05) rectangle (7.95, 7.95);
\path [fill=gray] (7.95,7.95) rectangle (8.05, 9.05); 
\path [fill=gray] (8.05,9.05) rectangle (6.95, 8.95);
\path [fill=gray] (7.05,9.05) rectangle (6.95, 7.95);
\path [fill=gray] (7.05,8.05) rectangle (5.95, 7.95);
\path [fill=gray] (5.95,7.95) rectangle (6.05, 9.05); 
\path [fill=gray] (6.05,9.05) rectangle (4.95, 8.95);
\path [fill=gray] (5.05,9.05) rectangle (4.95, 7.95);
\path [fill=gray] (5.05,8.05) rectangle (3.95, 7.95);
\path [fill=gray] (4.05,8.05) rectangle (3.95, 6.95);
\path [fill=gray] (4.05,7.05) rectangle (2.95, 6.95);
\path [fill=gray] (3.05,7.05) rectangle (2.95, 4.95);
\path [fill=gray] (3.05,5.05) rectangle (1.95, 4.95);
\path [fill=gray] (1.95,4.95) rectangle (2.05, 9.05);
\path [fill=gray] (2.05,9.05) rectangle (0.95, 8.95);
\path [fill=gray] (1.05,9.05) rectangle (0.95, 0.95);

%grid%
\foreach \x in {0,1,2,3,4,5,6,7,8,9}{
	\foreach \y in {0,1,2,3,4,5,6,7,8,9}{
		\draw  (\x+0.05,\y+0.05) rectangle (\x + 0.95, \y + 0.95);		
	}
}

%conenected component in green%
\foreach \x in {1,5,7} {
	\draw [fill] (\x +0.05,8.05) rectangle (\x +0.95,8.95);
}
\foreach \x in {1,4,5,6,7,8} {
	\draw  [fill] (\x +0.05,7.05) rectangle (\x +0.95,7.95);	
}
\foreach \x in {1,3,4,5,6} {
	\draw [fill] (\x +0.05,6.05) rectangle (\x +0.95,6.95);
}
\foreach \x in {1,3,4,5,6,8} {
	\draw [fill] (\x +0.05,5.05) rectangle (\x +0.95,5.95);
}
\foreach \x in {1,2,3,4,5,6,8} {
	\draw [fill] (\x +0.05,4.05) rectangle (\x +0.95,4.95);
}
\foreach \x in {1,2,5,6,8} {
	\draw [fill] (\x +0.05,3.05) rectangle (\x +0.95,3.95);
}
\foreach \x in {1,5,8} {
	\draw [fill] (\x +0.05,2.05) rectangle (\x +0.95,2.95);
}
\foreach \x in {1,2,3,4,5,6,7,8} {
	\draw [fill] (\x +0.05,1.05) rectangle (\x +0.95,1.95);
}

%labeling of important things%

\draw [thick] (7.01,7) circle [radius=0.15];
\draw [thick, white] (7.1,6.9) -- (10.5,5.7);
\draw [thick, dashed] (7.1,6.9) -- (10.5,5.7);
\draw (10.5,5.7) node[right]  {convex corner};

\draw [thick] (8.99,7.99) circle [radius=0.15];
\draw [thick, dashed] (9.1,8) -- (10.5,8.3);
\draw (10.5,8.3) node[right,text width = 2.8 cm]  {right-top corner of $T$};

\draw [thick] (7.99,8.99) circle [radius=0.15];
\draw [thick, dashed] (8.1,9) -- (9,10.5);
\draw (9,10.5) node[right]  {top-right corner of $T$};

\draw [thick] (1.01,1) circle [radius=0.15];
\draw [thick, dashed] (0.9,1) -- (-0.7,0.8);
\draw (-0.5,0.8) node[left, text width = 2.4cm]  {bottom-left \emph{and} left-bottom corner of $T$};

\draw [thick] (8.99,6.99) circle [radius=0.15];
\draw [thick, dashed] (8.99,6.99) -- (10.5,6.5);
\draw (10.5,6.5) node[right]  {concave corner};

\draw [white, thick] (3.5,2.8) -- (-0.5,2.5);
\draw [thick, dashed] (3.5,2.8) -- (-0.5,2.5);
\draw (-0.5,2.5) node[left]  {a hole in $T$};

\draw [white, thick] (5.5,1) -- (6,-0.5);
\draw [thick, dashed] (5.5,1) -- (6,-0.5);
\draw  (6,-0.5) node[below]  {outer border of $T$};

\draw [white, thick] (2.5,2) -- (2,-0.5);
\draw [thick, dashed] (2.5,2) -- (2,-0.5);
\draw (2,-0.5) node[below]  {inner border of $T$};

\draw [thick, dashed] (0.5, 8.5) -- (-0.5,8);
\draw [thick, dashed] (0.5, 7.5) -- (-0.5,8);
\draw (-0.5,8) node[left, text width = 2.8 cm] {\emph{adjacent} pieces};

\draw [thick] (9.5,1) ellipse (0.4cm and 0.15cm);
\draw [thick, dashed] (9.9, 1) -- (10.5,1.2);
\draw (10.5,1.2) node[right, text width=2.8cm]  {an edge containing a pair of \emph{adjacent} half-edges};

\draw [thick, dashed] (9.5, 0.05) -- (9.7, -0.5);
\draw (9.7, -0.5) node[below right, text width=2.8cm]  {an edge containing only one half-edge};

\draw  [thick, white] (5.5, 6.5) -- (6.5, 11);
\draw  [thick, dashed] (5.5, 6.5) -- (6.5, 11);
\draw (6.5, 11) node[above, text width=3.5cm]  {connected set $T \subset V$};

\end{tikzpicture}
\end{center}
\caption{A $10 \times 10$ grid $G$ with a connected set $T$ $\subset V$}
\label{fig:defs}
\end{figure}

Our first lemma states some basic observations about connected sets. 

\begin{lemma}\label{lemma:observations}
Let $T$ be a connected set of a given $n \times n$ grid $G$ and let $R$ be the bounding rectangle of $T$. Then the following holds:
\begin{enumerate}[$(i)$]
\item The outer border of $T$ contains exactly four more concave corners than convex corners.
\item Consider the left side of $R$ and the part of $T$'s border corresponding to this side, i.e. the border part between the top-left and bottom-left corner (in counter-clockwise direction, see figure \ref{fig:Lemma}). Then this part of $T$'s border contains the same number of concave and convex corners.  
\end{enumerate}
\end{lemma}

\begin{proof}
\noindent $(i)$: We will prove this statement by induction on the length of the outer border. As induction basis, note that a connected set consisting of one piece has an outer border with 4 concave corners and no convex corner. For the induction step, let $T$ be a connected set with an outer border of size $k$. By induction assumption we know that the outer border of any connected structure with border size less than $k$ contains four more concave than convex corners. 

Let us consider the bounding rectangle $R$ of $T$. Whenever the border of $T$ `leaves' the border of $R$ and `comes back again', an indentation is formed (see Figure \ref{fig:Lemma}). If this leaving and coming back happens on the same side of $R$ (see Figure \ref{fig:Lemma}, indentation $I$), we know by induction that the border of the indentation contains four more concave than convex corners. Note that all convex corners of the indentation border are concave corners of the border of $T$ and vice versa, except for the two corners incident to the border of $R$, which are convex corners of $T$ and of the indentation. Hence, there are as many concave as convex corners in the border part between the two corners where the border of $T$ leaves the border of $R$ and comes back. We can thus fill in the indentation to get a set $T'$ with the same difference between convex and concave corners as $T$ but with a smaller outer border. The claim thus follows from the induction hypothesis.

If there is an indentation that leaves the border of $R$ on one side and comes back on the next side of $R$ (see Figure \ref{fig:Lemma}, indentation $I'$), then a similar argument shows that the border of $T$ between the two corners where it leaves the border of $R$ and comes back contains one more concave corner -- and filling in the indentation concludes the proof similar as before.

The only case remaining is that the border of $T$ contains no indentation at all, implying that $T$ and $R$ are identical and thus the claim trivially holds.

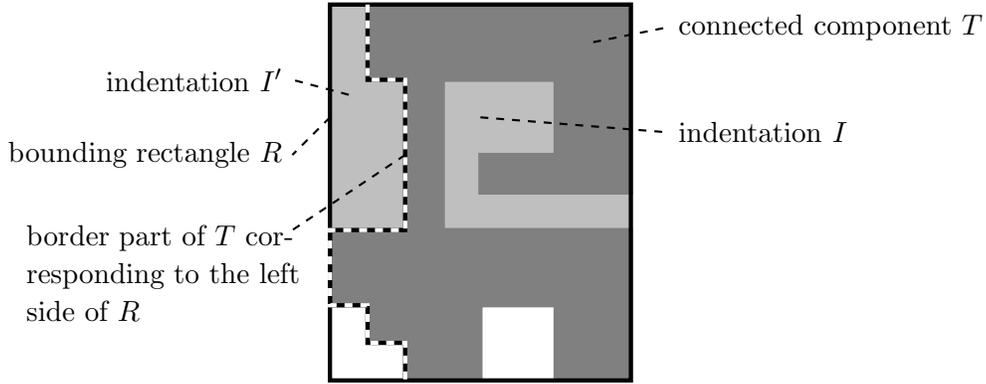
\begin{figure}
\begin{center}
\begin{tikzpicture}

%labeling of important things%

\draw [dashed, thick](0,3.5) -- (-0.5,3);
\draw (-0.5,3) node[left]  {bounding rectangle $R$};

%indentation I%
\draw [ultra thick, fill=lightgray, lightgray] (4,2) -- (1.5,2) -- (1.5,4) -- (3,4) -- 
					(3,3) -- (2, 3) --  (2,2.5) -- (4, 2.5) -- (4,2);

%indentation I'%
\draw [ultra thick, fill=lightgray, lightgray] (0.5,5) -- (0.5,4) -- (1,4) --  (1,2) -- (0, 2)--(0,5) -- (0.5,5);

%connected component in gray%
\draw [fill=gray, ultra thick, gray] (1,0) -- (2,0) -- (2,1) -- (3,1) -- 
 	(3,0) -- (4,0) -- (4,1) -- (4,2) -- (1.5, 2) -- (1.5,4) -- (3,4)
 	 -- (3,3) -- (2, 3) --  (2,2.5) -- (4, 2.5)  -- (4,5) --
 	 (0.5, 5)-- (0.5,4) -- (1,4) --  (1,2) -- (0, 2) -- (0,1) --(0.5,1) -- (0.5,0.5) --(1,0.5) -- (1,0);

%Enclosing Rectangle in green%
  \draw [ultra thick] (0,0) rectangle (4,5);

%left side of T%
\draw [ultra thick]  (0.5, 5)-- (0.5,4) -- (1,4) --  (1,2) -- (0, 2) -- (0,1) --(0.5,1) -- (0.5,0.5) --(1,0.5) -- (1,0);
\draw [ultra thick, dashed, white]  (0.5, 5)-- (0.5,4) -- (1,4) --  (1,2) -- (0, 2) -- (0,1) --(0.5,1) -- (0.5,0.5) --(1,0.5) -- (1,0);

\draw [dashed, thick] (3.5,4.5) -- (4.5,4.7);
\draw (4.5,4.7) node[right]  {connected component $T$};

\draw  [dashed, thick](0.25,3.8) -- (-0.5, 4);
\draw (-0.5, 4) node[left]  {indentation $I'$};

\draw [dashed, thick] (2,3.5) -- (4.5,3.3);
\draw (4.5,3.3) node[right]  {indentation $I$};

\draw  [thick, dashed](1,3) -- (-0.5,2);
\draw (0.1,2.2) node[below left, text width=4cm]  {border part of $T$ corresponding to the left side of $R$};	 
 	 
\end{tikzpicture}
\end{center}
\label{fig:Lemma}
\caption{Indentations}
\end{figure}

\noindent $(ii)$: We call the border part of $T$ between the top-left and bottom-left corner the \emph{left boundary} and the rest of the border of $T$ the \emph{right boundary}. If the left boundary of $T$ is not just a straight line, by filling all the indentations along the right boundary, we get a set $T'$ with four more concave than convex corners (by $(i)$). But, since the right boundary of $T'$ coincides completely with the bounding rectangle $R$, it contains two concave and no convex corner. Thus, since the top-left and the bottom-left corner are concave corners as well, it holds that the left boundary contains the same number of concave and convex corners.  

Observe that this statement is also true for the border parts of $T$ corresponding to the right, upper and lower side of $R$.
\end{proof}

\section{Proof of Theorem \ref{thm:Jigsaw}, $q = o(n)$}

The non-uniqueness comes essentially from the birthday paradox, cf.~\cite{MosselPC}. We add the short proof for completeness. The main idea is that w.h.p. there exist two pieces in $V$ that have the exact same 4-tuple of colours assigned to their half-edges. These pieces can be interchanged resulting in at least two different valid reconstructions. 

Note that a random colouring function $g$ assigns the same colour to every two adjacent half-edges, thus we do not have the property that 4-tuples of colours assigned to pieces are chosen independently from $Q^4$. One can circumvent this by only considering every other piece. More precisely, let 
$$
	V' := \{v_{i,j} \in V \mid i+j \mbox{ is even}\}.
$$
Note that $|V'| = \lceil n^2 / 2 \rceil$ and that $V'$ is a set of pairwise non-adjacent pieces  (it consists of all `black' fields of an $n \times n$ chessboard). Since the pieces in $V'$ are pairwise non-adjacent, it is easy to see that the corresponding 4-tuples of colours are mutually independent. Since the colouring of the half-edges of the pieces in $V'$ is mutually independent, we can assume that the pieces receive their colour one after the other independently and uniformly at random. Conditioning on the event that the there are no two identical pieces among the first $k$ pieces, the probability that the $(k+1)$st piece is identical to one of the first $k$ pieces is $(1 - \frac{k}{q^4})$. Therefore, we have that 
$$\Pr[\mbox{no two identical pieces}] = \prod_{k=1}^{\frac{n^2}{2}-1} (1 - \frac{k}{q^4}) \le e^{-\sum_{k=1}^{\frac{n^2}{2}-1} \frac{k}{q^4}} =  e^{-\frac{n^4-2n^2}{8q^4}} = o(1),$$
where we used that $1-x \leq e^{-x}$ and $q =o(n) $. This means that with probability $1 - o(1)$ there are two pieces in $V'$ which can be swapped and thus the puzzle is w.h.p. not uniquely reconstructable. Note that this argument still applies if we do not allow rotations.

%Note that the factor 4 in the probability  $(1 - \frac{4k}{q^4})$ comes from the rotation of the pieces, since any chosen 4-tuple of colours makes all 4 rotations to a `bad' set of colours.

\section{Proof of Theorem \ref{thm:Jigsaw}, $q \ge n^{1 + \varepsilon}$}

A straightforward first moment argument shows that we do not expect two identical pieces. Indeed, let $X_{(i,j),(i',j')}$ be the indicator variable for the event that $v_{i,j}$ and $v_{i',j'}$ are identical. Then the probability that this variable is one is at most $4/q^4$ if the two pieces are not adjacent in $G$ (here the constant $4$ takes care of the rotations) and at most $4/q^3$ if   $v_{i,j}$ and $v_{i',j'}$ are adjacent. As we have less than $n^4$ such variables and only $O(n^2)$ adjacent pairs, the expected number of identical pieces is $o(1)$. From Markov's inequality it thus follows that w.h.p. we do not have any pair of identical pieces. 

Clearly, this argument shows that we cannot obtain a valid reconstruction by swapping only two pieces. %The core of our proof idea is to show that the general argument for `bigger components' is very similar to the above proof for single pieces.
Before we show that a similar statement holds for `bigger components', we need some more definitions. Two half-edges which are adjacent in $G$ are called \emph{partners}. Let $H = G(\phi, \Pi)$ be a reconstruction. An edge $e \in E_H$ is called an \emph{original edge} if its two half-edges are partners, i.e. if $e \in E_G$ as well. Else, we call it a \emph{new edge}. Furthermore, let $T \subset V$ be a connected set in $H$. We call $T$ \emph{stable} if and only if all edges in $E_H(T)$ are original edges and all edges in $E_H(T, V \setminus T)$ are new edges. On the one hand this implies that $T$ is a connected set in $G$ as well (but not necessarily at the same place, and it might be rotated) and, on the other hand, that all pieces along the border of $T$ in the reconstruction $H$ are either distinct from the corresponding piece in $G$ or are at least rotated.

%Let us first give a rough outline of our proof idea for the upper bound. Intuitively, we will provide a set of properties that any stable set $T$ in a valid reconstruction $H^{\Pi}_{\phi}$ may {\em not } have. Using these properties, we then argue that a stable set $T$ cannot be bigger than a constant depending on $\eps$, and therewith we can show that w.h.p. $G$ itself is the only valid reconstruction.

We show that certain `undesired' structures, which can be thought of as a part of the border of a stable set $T$, w.h.p. cannot be a part of a valid reconstruction. %This means that if we take some connected sets from the puzzle and try to build a certain structure with their pieces, however we put the pieces together to get this structure, it is w.h.p. not valid. 
By defining these structures appropriately, this will in turn allow us to conclude that the puzzle has a unique reconstruction. For example, assume that $T$ is a $1\times 4$ subgrid of $G$ and assume that we attach to one of its longer sides another $1\times 4$ subgrid, different than the one which appears in $G$. Then the probability that all four edges are monochromatic is $1/q^4$.  Thus, since there are at most $O(n^2)$ subgrids of constant size in $G$, the expected number of such valid configurations is bounded by $O(n^2\cdot n^2 / q^4)=o(1)$. By Markov's inequality, such valid configurations do not appear.  %Similarly, if we assume that $T$ is a $1\times \lceil 2/\eps\rceil$ subgrid and the attached pieces along one side all come from different locations, then the probability that all edges match is $1/q^{2\lceil 2/\eps \rceil-1}$ (the $\lceil 2/\eps\rceil$ half-edges with $T$ and the $\lceil 2/\eps\rceil-1$ edges between the attached pieces). So the expected number of such configurations is $O(n^2\cdot (n^2)^{\lceil 2/\eps\rceil} / q^{2\lceil2/\eps\rceil-1})=o(1)$. 

%Note, however, that these are not all the cases we have to consider. We also need to consider the 'mixed' cases, where some of the pieces that we attach are connected in original puzzle, while others are not. %In addition, we have to consider the case that we use two adjacent pieces in the original puzzle at non-adjacent positions in the reconstruction. We thus have to be very careful in order not to sample dependent half-edges multiple times. 
%Nevertheless, for the proof of Theorem~\ref{thm:Jigsaw} it suffices to just use first moment arguments as described in the two examples above. That is, we show that the expected number of certain 'undesired` valid configurations is $o(1)$. From Markov's inequality we can thus deduce that with probability $1-o(1)$ a random colouring of the original puzzle is such that w.h.p.\ none of these  'undesired` configurations are valid. By defining  the notion of 'undesired` configurations appropriately, this will then in turn allow us to conclude that the puzzle has a unique reconstruction.

%As we will see, the difficult part for these first moment arguments is
%calculating the probabilities that a given configuration is valid. To this end, 

Note that, in general, it is not true that we get a factor of $1/q$ for each new edge. Consider, for example, two edges $(a,a')$ and $(b,b')$ that are rearranged to edges $(a,b')$ and $(b,a')$. Clearly, the probability that they are both monochromatic is just $1/q$ and not $1/q^2$. Before we describe the undesired configurations, we first give two propositions to handle such situations.
%Our key tool for handling such situations are the following two propositions. % about sets of new edges, which give us upper bounds for the probability for the validity of the above mentioned sampling process.
The first proposition handles the case when we can order the given new edges in such a way that we can bypass the dependencies of the colouring of the half-edges, whereas the second proposition handles the general case.

The idea of the proofs of these propositions is the so-called principle of deferred decision. That is, we do not reveal the colouring of the original edges all at once, but instead we reveal them whenever we need to know whether a new edge is monochromatic. Clearly, this implies that whenever we have so far revealed at most one of the half edges of a new edge, the probability that the edge is monochromatic is exactly $1/q$. Of course, %using that this probability is $1/q$ in turn means that we have to
revealing the colour of the other half-edge means that we also reveal the colour of the partner of that half-edge. Therefore, the order in which we reveal colours plays an essential part in obtaining good probability bounds.

\begin{proposition}\label{prop:GoodEdges} 
Let  $U = \{e_1, \ldots, e_m \}$ be a set of new edges and, for all $i \in [m]$, let $h_i^1$ and $h_i^2$ be the two half-edges of $e_i$. Assume that the colours of all half-edges are still unknown and that at most one of the partners of $h_i^1$ and $h_i^2$ is an element of $\bigcup_{k=1}^{i-1} e_k$, for every $i \in [m]$. Then
$$\Pr[\mbox{all edges in } U \mbox{ are monochromatic}] = \frac{1}{q^{m}}.$$
\end{proposition}

\begin{proof}
We consider the edges~$e_1,\ldots,e_m$ one after the other, in this order. By the assumption
that at most one of the partners of $h_i^1$ and $h_i^2$ is an element of $\bigcup_{k=1}^{i-1} e_k$, when we consider $e_i$, the colour of at least one of it's half-edges is still unknown. Thus, the probability that $e_i$ is monochromatic is $1/q$. As all these events are independent we get that the probability that all the  edges of $U$ are monochromatic is equal to $1/q^m$.
\end{proof}

In the general case, we will not be able to order the edges so that we can bypass the dependencies of the colouring. Moreover, as we will later apply Propositions~\ref{prop:GoodEdges} and~\ref{prop:NewEdgePairs} repeatedly (and to different sets of edges),
we will not be able to always assume that the colours of all half-edges are still unknown. In Proposition~\ref{prop:NewEdgePairs} we thus allow that the colour of some half-edges are already known.

\begin{proposition}\label{prop:NewEdgePairs} 
Let  $U = \{e_1, \ldots, e_m \}$ be a set of new edges and, for all $i \in [m]$, let $h_i^1$ and $h_i^2$ be the two half-edges of $e_i$. If for every $i \in [m]$ at most one of the colours of $h_i^1$ and $h_i^2$ is already known, then 
$$\Pr[\mbox{all edges in } U \mbox{ are monochromatic}] \leq \frac{1}{q^{\lceil m/2 \rceil}}.$$
\end{proposition}

\begin{proof}
We consider the edges in $U$ one by one, but not necessarily in the order $e_1, \ldots, e_m$.
We call an edge \emph{open} if the colour of both its half-edges is still unknown, \emph{critical} if the colour of exactly one of its half-edges was already revealed, and {\em closed} otherwise. By  assumption, all edges in $U$ are either
open or critical in the beginning. We argue that we can consider $\lceil m/2 \rceil$ edges of $U$ in such an order that in each step we consider an open or critical edge (and thus get a factor of $1/q$) and only the considered edge and at most one additional edge change their status to closed. Clearly, this will imply the desired bound.

If there is at least one critical edge $e_i$ in $U$, we consider it and reveal the colour of its so-far uncoloured half-edge. The probability that $e_i$ is monochromatic is thus equal to $1/q$. By revealing the colour of the half-edge we also reveal the colour of its partner, which may imply that one (but only one!) more edge in $U$ changes its status to closed. 

Otherwise, if $U$ contains only open or closed edges and at least one open edge, then revealing the colours of both half-edges of an open edge again implies that the probability that this edge is monochromatic is $1/q$ and, again, either at most one more edge can change its status to closed or at most two can change the status to critical.
\end{proof}

The proof of Theorem \ref{thm:Jigsaw} relies on the concept of \emph{configurations}. We start by describing the general setup. Let $C_1, \ldots, C_m$ be disjoint sets which are connected in $G$ such that $C_i \cup C_j$ is not a connected set (for distinct $i, j \in [m]$) and $|\bigcup_{i=1}^{m} C_i | \leq K(\varepsilon)^2$, where $K(\varepsilon)$ is a constant which we define later. In the cases we consider, we always reassemble the pieces into a connected set $S \cup U = \bigcup_{i=1}^{m} C_i$ such that
\begin{itemize}
	\item $S \subseteq C_1$ is a connected set in $G$, and 
	\item all edges between $S$ and $U$ are new edges. 
\end{itemize}
We call an $(m+1)$-tuple $(C_1, \ldots, C_m, \mathtt{rule})$ a \emph{configuration}, where  $\mathtt{rule}$ tells us how to rearrange and rotate the pieces of $C_1, \ldots, C_m$. Thus $\mathtt{rule}$ consists of two functions $\phi'$ and $\Pi'$, where $\phi'  \colon \bigcup_{i=1}^{m} C_i \rightarrow K(\varepsilon) \times K(\varepsilon)$ is an injective function specifying the rearrangement of the pieces and $\Pi'$ assigns any piece of $\bigcup_{i=1}^{m} C_i$ a cyclic shift. 

In the following, we will consider five \emph{types} of configurations, where the type of the configuration indicates the `shape' of the sets $S$ and $U$. Figure~\ref{fig:configurations} illustrates the first four types. In this figure the set $S$ is composed of black pieces and the set $U$ is composed of certain pieces next to $S$ (we will make this precise for each type below).
Intuitively, one can think of $S$ as being a part of the border of a stable set and $U$ as a set of pieces along this border that all define new edges with $S$. Note that for some types, namely $\mathtt{straightline}$, $\mathtt{convexcorners}$ and $\mathtt{subsquare}$, we additionally impose a lower bound on the size of $S$ and/or $U$. Moreover, $\mathtt{subsquare}$ imposes some further technical conditions on induced stable sets, which will be discussed later. We call a configuration $(C_1, \ldots, C_m, \mathtt{rule})$ \emph{valid} if all edges inside $S \cup U$ are monochromatic. We denote by $X_{(C_1, \ldots, C_m, \mathtt{rule})}$ the indicator random variable for this event.

Our aim is to show that the expected number of valid configurations $(C_1, \ldots, C_m, \mathtt{rule})$ of all five types in a random jigsaw puzzle is $o(1)$, since then Markov's inequality implies that w.h.p. none of these configurations can be found in a valid reconstruction. As we assumed that $|S\cup U| =  |\bigcup_{i=1}^{m} C_i| = O(1)$, the number of tuples $(C_1,\ldots,C_m)$ is bounded by $O(n^{2m})$ and for each such tuple (and each type) we have only $O(1)$ different ways to reassemble the pieces into $S \cup U$.  %($O(1)$ ways to choose the precise `shape' of the set $S \cup U$ according to its type and $O(1)$ different $\mathtt{rules}$ to arrange the pieces in each such set $S \cup U$). That is, the number of configurations of each type is bounded by $O(n^{2m})$. %Observe that, as the number of pieces is constant, there are only constantly many ways to reassemble the pieces into a connected set $S \cup U$. Thus, since each $C_i$ is of constant size as well, the number of configurations $(C_1, \ldots, C_m, \mathtt{rule})$ is bounded by , where the hidden constant depends only on $\eps$.
By linearity of expectation it thus suffices to show that  $\Pr[X_{(C_1, \ldots, C_m, \mathtt{rule})}=1] = o(n^{-2m})$. 
This is what we will do now, details depending on the type of the configuration that we consider. Since the problem is monotone in the number of colours (if we recolour all edges of a fixed colour uniformly at random with the remaining colours, then this increases the likelihood of unique reconstructability), we can assume that $\eps < 1/4$.

%Note that the probability that a given configuration is valid is not identical for all configurations, as the probability whether edges between pieces in $S \cup U$ are properly coloured depends on whether this edge is an original or a new edge.  

\begin{figure}[!ht]
\begin{center}
\begin{tikzpicture}

%straightline%

%Border part in red%
 \draw [fill=gray, gray]  (-4, 2.05) rectangle (-1, 1.95);

\draw [dashed](-1, 2) -- (0, 2.5); 
\draw (0, 2.5) node[above]  {set of new edges};

%Set U %
\foreach \x in {-4,-3.5,-3,-2.5,-2,-1.5} {
\draw   (\x+0.05,1.55) rectangle (\x + 0.45, 1.95);
 }

%Set S in green% 	 
\foreach \x in {-4,-3.5,-3,-2.5,-2,-1.5} {
\draw  [fill] (\x+0.05,2.45) rectangle (\x + 0.45, 2.05);
 }

\draw [dashed](-2.25, 2.25) -- (-2.5, 3); 
\draw (-2.5, 3) node[above]  {set $S$};

\draw [dashed](-2.25, 1.5) -- (-2.5, 1); 
\draw (-2.5, 1) node[below]  {set $U$};

%convexcorners%

%Border part in red%
 \draw [fill=gray, gray] (6,0.95) -- (5.45, 0.95) -- (5.45,1.45) -- (4.95,1.45)
 	--	(4.95,1.95) -- (3.95,1.95) -- (3.95,2.55) -- (5.45, 2.55) -- (5.45,3.05) -- (6,3.05)
 	-- (6,2.95) -- (5.55,2.95) -- (5.55, 2.45) -- (4.05, 2.45) 
 	-- (4.05, 2.05) -- (5.05, 2.05) -- (5.05, 1.55) -- (5.55, 1.55)
 	-- (5.55, 1.05) -- (6,1.05) -- (6,0.95) ;

%Set S in green%
\foreach \x in {5,5.5} {
\draw  [fill] (\x+0.05, 3.05) rectangle (\x + 0.45, 3.45);
 }
\foreach \x in {5, 4.5, 4, 3.5} {
\draw  [fill] (\x+0.05, 2.55) rectangle (\x + 0.45, 2.95);
 }
\foreach \x in {3.5} {
\draw  [fill] (\x+0.05, 2.45) rectangle (\x + 0.45, 2.05);
 }
\foreach \x in {4.5, 4, 3.5} {
\draw  [fill] (\x+0.05, 1.55) rectangle (\x + 0.45, 1.95);
 }
\foreach \x in {5, 4.5} {
\draw  [fill] (\x+0.05, 1.05) rectangle (\x + 0.45, 1.45);
 }
 \foreach \x in {5, 5.5} {
\draw  [fill] (\x+0.05, 0.55) rectangle (\x + 0.45, 0.95);
 }

%Set U%
\draw   (5.55,2.55) rectangle (5.95, 2.95);
\draw  (4.05,2.05) rectangle (4.45, 2.45);
\draw  (5.05,1.55) rectangle (5.45, 1.95);
\draw   (5.55,1.05) rectangle (5.95, 1.45);

\draw [dashed](6,1.5) -- (6.5, 2); 
\draw [dashed](5.5,1.75) -- (6.5, 2); 
\draw [dashed](4.5,2.25) -- (6.5, 2); 
\draw [dashed](6,2.5) -- (6.5, 2); 
\draw (6.5, 2) node[right]  {set $U$};

\draw [dashed](5.5,1) -- (4.3, 1); 
\draw (4.3,1) node[left]  {set of new edges};

\draw [dashed](3.75, 3) -- (3.75, 3.35); 
\draw (3.75,3.35) node[above]  {set $S$};

%hole%

%Border part in red%
 \draw [fill=gray, gray] (-2.05,-4.45) -- (-2.45, -4.45) -- (-2.45,-3.95) -- (-3.45,-3.95)
 	--	(-3.45,-3.05) -- (-2.45,-3.05) -- (-2.45,-2.55) -- (-1.55, -2.55) -- (-1.55,-2.95) -- (-2.05,-2.95)
 	-- (-2.05,-4.55) -- (-1.95,-4.55) -- (-1.95, -3.05) -- (-1.45, -3.05) -- (-1.45, -2.45) 
 	-- (-2.55, -2.45) -- (-2.55, -2.95) -- (-3.55, -2.95) -- (-3.55, -4.05)
 	-- (-2.55, -4.05) -- (-2.55,-4.55) -- (-2.05,-4.55) -- (-2.05,-4.45) ;

%Set S in green%
\foreach \x in {-3,-2.5, -2, -1.5} {
\draw  [fill] (\x+0.05, -2.05) rectangle (\x + 0.45, -2.45);
 }
\foreach \x in {-1.5,-4, -3.5, -3} {
\draw  [fill] (\x+0.05, -2.55) rectangle (\x + 0.45, -2.95);
 }
\foreach \x in {-1.5,-2, -4} { 
\draw  [fill] (\x + 0.05, -3.05) rectangle (\x + 0.45, -3.45);
}
\foreach \x in {-4, -2} {
\draw  [fill] (\x+0.05, -3.55) rectangle (\x+0.45, -3.95);
}
\foreach \x in {-4, -3.5, -3, -2} {
\draw  [fill] (\x+0.05, -4.05) rectangle (\x + 0.45, -4.45);
 }
 \foreach \x in {-3,-2.5, -2} {
\draw  [fill] (\x+0.05, -4.55) rectangle (\x + 0.45, -4.95);
 }

%Set U%
\foreach \x in {-2, -2.5} {
\draw (\x+0.05, -2.55) rectangle (\x + 0.45, -2.95);
 }
\foreach \x in {-3.5,-3, -2.5} { 
\draw   (\x + 0.05, -3.05) rectangle (\x + 0.45, -3.45);
}
\foreach \x in {-3.5,-3, -2.5} {
\draw  (\x+0.05, -3.55) rectangle (\x+0.45, -3.95);
}
\foreach \x in {-2.5} {
\draw   (\x+0.05, -4.05) rectangle (\x + 0.45, -4.45);
 }
 
  \draw [white](-2.75, -2.75) -- (-3.5, -2); 
 \draw [dashed](-2.75, -2.75) -- (-3.5, -2); 
\draw (-3.5,-2) node[above]  {set $S$};

\draw [white](-2.5,-3.5) -- (-0.5, -3.5); 
\draw [dashed](-2.5,-3.5) -- (-0.5, -3.5); 
\draw (-0.5, -3.5) node[right]  {set $U$}; 

\draw [white] (-3,-4) -- (-4, -5); 
\draw [dashed] (-3,-4) -- (-4, -5); 
\draw (-4,-5) node[below]  {set of new edges};

%indentation%

%Border part in red%
 \draw [fill=gray, gray] (6,-4.55) -- (5.45, -4.55) -- (5.45,-4.05) -- (4.95,-4.05)
 	--	(4.95,-3.55) -- (3.95,-3.55) -- (3.95,-2.95) -- (5.45, -2.95) -- (5.45,-2.45) -- (6,-2.45)
 	-- (6,-2.55) -- (5.55,-2.55) -- (5.55, -3.05) -- (4.05, -3.05) 
 	-- (4.05, -3.45) -- (5.05, -3.45) -- (5.05, -3.95) -- (5.55, -3.95)
 	-- (5.55, -4.45) -- (6, -4.45) -- (6, -4.55) ;
 	
%Border of the grid in blue%
 \draw [fill=darkgray, darkgray] (5.95,-2) rectangle (6.05, -5) ;
 \draw  [white](6,-3) -- (6.5, -2.5); 
  \draw  [dashed](6,-3) -- (6.5, -2.5);  
\draw (6.5, -2.5) node[right]  {border of the grid}; 
 
%Set S in green%
\foreach \x in {5,5.5} {
\draw  [fill] (\x+0.05, -2.05) rectangle (\x + 0.45, -2.45);
 }
\foreach \x in {5, 4.5, 4, 3.5} {
\draw  [fill] (\x+0.05, -2.95) rectangle (\x + 0.45, -2.55);
 }
\foreach \x in {3.5} {
\draw  [fill] (\x+0.05, -3.05) rectangle (\x + 0.45, -3.45);
 }
\foreach \x in {4.5, 4, 3.5} {
\draw  [fill] (\x+0.05, -3.95) rectangle (\x + 0.45, -3.55);
 }
\foreach \x in {5, 4.5} {
\draw  [fill] (\x+0.05, -4.45) rectangle (\x + 0.45, -4.05);
 }
 \foreach \x in {5, 5.5} {
\draw  [fill] (\x+0.05, -4.95) rectangle (\x + 0.45, -4.55);
 }

%Set U%
\foreach \x in {5.5} {
\draw  (\x+0.05, -2.95) rectangle (\x + 0.45, -2.55);
 }
\foreach \x in {4,4.5,5,5.5} {
\draw  (\x+0.05, -3.05) rectangle (\x + 0.45, -3.45);
 }
\foreach \x in {5,5.5} {
\draw   (\x+0.05, -3.95) rectangle (\x + 0.45, -3.55);
 }
\foreach \x in {5.5} {
\draw  (\x+0.05, -4.45) rectangle (\x + 0.45, -4.05);
 }

\draw [white](5.25, -2.75) -- (4, -2); 
\draw [dashed](5.25, -2.75) -- (4, -2); 
\draw (4,-2) node[above]  {set $S$};

\draw [white] (5.5,-3.5) -- (6.5, -3.5);  
\draw [dashed] (5.5,-3.5) -- (6.5, -3.5);  
\draw (6.5, -3.5) node[right]  {set $U$}; 

\draw  [white](4.5,-3.5) -- (3.5, -4.5); 
\draw  [dashed](4.5,-3.5) -- (3.5, -4.5); 
\draw (3.5,-4.5) node[below]  {set of new edges};

%labbeling the three pictures%
\draw (-2.5,0) node[below, text width=4cm]  {$(C_1, \ldots, C_m,  \mathtt{rule})$ of the type $ \mathtt{straightline}$}; 

\draw (4.5,0) node[below, text width=4.3cm]  {$(C_1, \ldots, C_m,  \mathtt{rule})$ of the type $ \mathtt{convexcorners}$}; 

\draw (-2.5, -6) node[below, text width=4cm]  { $(C_1, \ldots, C_m,  \mathtt{rule})$ of the type $ \mathtt{hole}$}; 

\draw (4.5, -6) node[below, text width=4.3cm]  { $(C_1, \ldots, C_m,  \mathtt{rule})$ of the type $ \mathtt{indentation}$}; 

\end{tikzpicture}
\end{center}
\caption{Illustration of different types of configurations $(C_1, \ldots, C_m,  \mathtt{rule})$; the black squares depict the pieces in $S$, while the white squares depict the pieces in $U$}
\label{fig:configurations}
\end{figure}
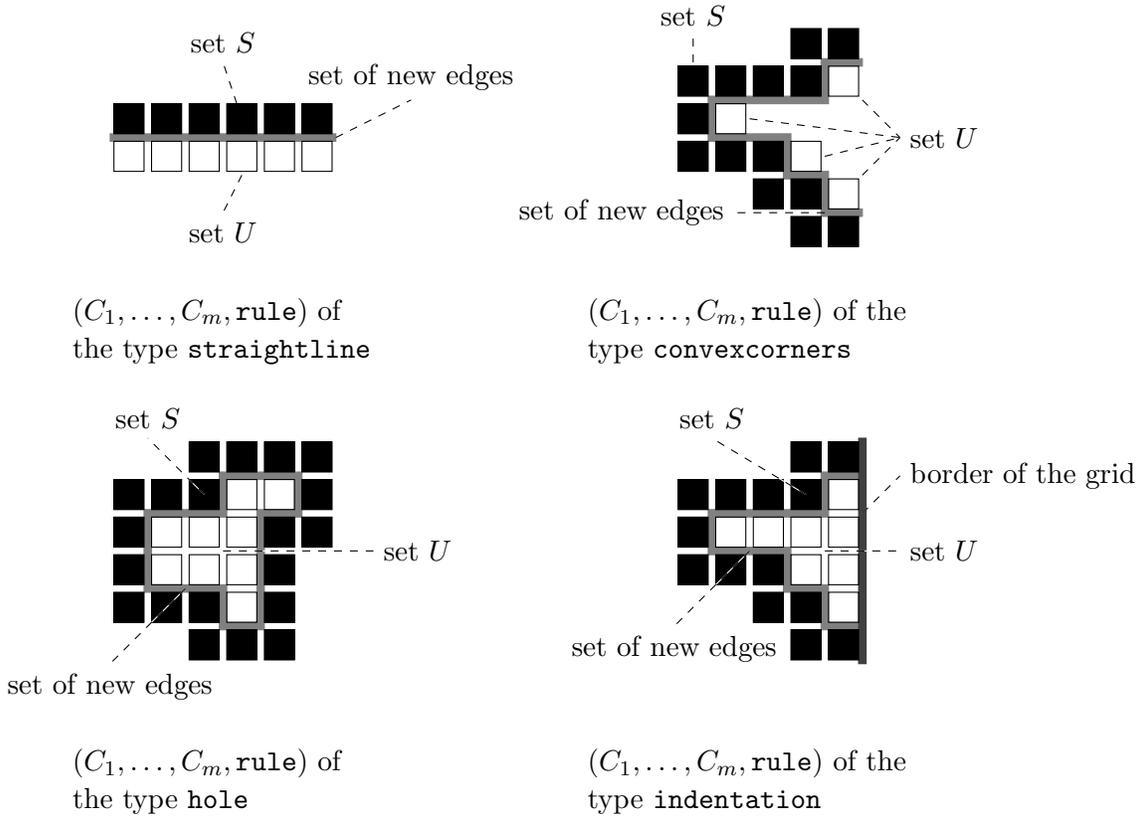

\noindent \textbf{$\mathtt{straightline}$:}
A configuration $(C_1, \ldots, C_m,  \mathtt{rule})$ of the type $\mathtt{straightline}$ is defined as illustrated in the top left part of Figure \ref{fig:configurations}. The key characteristics are:
\begin{itemize}
	\item $S$ is a connected \emph{row} of length $\ell(\varepsilon) := \lceil 3(1 +1 / \eps) \rceil$, and 
	\item $U$ is the row along the lower side of $S$. 
\end{itemize}
We want to bound the probability that all edges within $S\cup U$ are monochromatic.  To do this we apply Propositions~\ref{prop:GoodEdges} and \ref{prop:NewEdgePairs}  to certain sets of (new) edges within $S\cup U$. Recall that, by assumption, all edges between $S$ and $U$ are new.

We first consider the leftmost piece in $U$ from every component $C_i$, $2\le i\le m$. All but at most one of these $m-1$ pieces have a left neighbour in $U$. Let $W \subseteq U$ denote the set of  $m-2$ of such pieces and let ${\cal L}$ denote the set of edges between the pieces in $W$ and their left neighbour in $U$ together with the edges of these pieces to the set $S$ (i.e. we look at the edges going left and up from each piece in $W$). Then $|{\cal L}| = 2(m-2)$ and  one easily checks that Proposition~\ref{prop:GoodEdges}  implies that the probability that all these edges are valid is $ q^{-|{\cal L}|}$.

Next consider the set ${\cal L}'$ of edges between the pieces of $U \setminus W$ and $S$. Note that $|{\cal L}'| = |U| - (m-2)$. Since $W$ does not contain any piece from $C_1$, by colouring all edges in ${\cal L}$ (and thus the partners of the corresponding half-edges) we have that every edge in ${\cal L}'$ still has at least one half-edge for which its colour has not yet been revealed. Indeed, the only way the lower half-edge of a piece in $S$ adjacent to some piece in $U \setminus W$ could have already received a colour is if its partner half-edge belongs to some edge in $\cal L$. However, the piece containing such a half-edge then belongs to $C_1 \setminus S$ and, as $W$ does not contain any piece from $C_1$, we have not yet revealed half-edges of such pieces. Proposition~\ref{prop:NewEdgePairs} thus implies that the probability that all these edges are valid is at most $q^{-|{\cal L}'|/2}$. Therefore, as $q = n^{1+\eps}$ and $|U| = |S| =  \lceil 3(1 +1 / \eps) \rceil$ we obtain
\begin{align*}
\Pr[X_{(C_1, \ldots, C_m, \mathtt{rule})}=1] &\leq q^{-|{\cal L}|} \cdot q^{-|{\cal L}'|/2} \\
 &= q^{-2(m-2) - (|U|- (m-2))/2} \\
&= q^{3 -2m - (|U| - m)/2} \\
&\leq \begin{cases}
	 q^{3 -  2m - |U| /4 } , &\text{if }  m \leq  |U| /2\\
	 q^{3  - 2m}, &\text{if } m >   |U| /2\\
\end{cases} \\
&\leq \begin{cases}
	 n^{- 2m + 3(1+\eps) -    3(1 + \eps) (1+1/ \eps) /4 } , &\text{if }  m \leq |U| /2\\
	 n^{-2m + 3(1+\eps)  - 2\eps m}, &\text{if } m >  |U| /2\\
\end{cases} \\
&= o(n^{-2m}),
\end{align*}
where the last step follows from the assumption that $\eps < 1/4$ and since $2\eps m >3(1+\eps)$ if  $m > |U| / 2 =  \lceil 3(1 +1 / \eps) \rceil/2$ .
\medskip

\noindent \textbf{$\mathtt{convex corners}$:}
A configuration $(C_1, \ldots, C_m,  \mathtt{rule})$ of the type $ \mathtt{convexcorners}$ is defined  as illustrated in the top right part of Figure \ref{fig:configurations}. The key characteristics are:
\begin{itemize}
	\item $S$ contains $2 \ell(\varepsilon)$ convex corners facing to the right, 
	\item $U$ is the set of pieces which are placed in these corners, and 
	\item $|S| \leq s(\eps) := 4 \ell(\varepsilon)^2$.
\end{itemize}
Since any piece can belong to at most two convex corners, we have $2 \ell(\varepsilon) \geq |U| \geq \ell(\varepsilon)$. We bound the probability that all edges within $S \cup U$ are monochromatic. %To do this we apply Propositions~\ref{prop:GoodEdges} and \ref{prop:NewEdgePairs}  to certain sets of (new) edges within $S \cup U$. Recall that, by assumption, all edges between $S$ and $U$ are new.

Let us first consider one piece in $U$ from every component $C_i$, $2\le i\le m$ and let $W \subseteq U$ be the set of such pieces. %Then  $|W| = m-1$ and  $W$ contains exactly one piece from every component $C_2,\ldots,C_m$. 
Since all these pieces are placed in a convex corner, each will have at least two edges going to the set $S$.
Let ${\cal L}$ denote the set of edges between the pieces in $W$ and $S$. Then  $|{\cal L}| \geq 2(m-1)$ and Proposition~\ref{prop:GoodEdges} implies that the probability that all these edges are valid is at most $q^{-|{\cal L}|}$.

Next consider the set ${\cal L}'$ of edges between the pieces of $U \setminus W$ and $S$.
Note that we have $|{\cal L}'| \geq 2 \left(|U|- (m-1) \right)$. Similarly as in the $\mathtt{straightline}$ case, since $W$ does not contain any piece from $C_1$, by fixing the colouring for all edges in ${\cal L}$ we have that every edge in ${\cal L}'$ still has one half-edge for which the colour has not yet been revealed. Proposition~\ref{prop:NewEdgePairs} thus implies that the probability that all these edges are valid is at most $q^{-|{\cal L}'|/2} = q^{-\left(|U|- (m-1) \right)}$. Therefore, as $q = n^{1+\eps}$ and $|U| /2 \geq \lceil 3(1 +1 / \eps) \rceil / 2 $,  by similar calculations as in the previous case we obtain
\begin{align*}
\Pr[X_{(C_1, \ldots, C_m, \mathtt{rule})}=1] &\leq q^{-|{\cal L}|} \cdot q^{-|{\cal L}'|/2} \\
 &= q^{-2(m-2) - (|U|- (m-1))} \\
&= q^{3 -2m - (|U| - m)} \\
&\leq \begin{cases}
	 q^{3 -  2m - |U| /2 } , &\text{if }  m \leq  |U| /2\\
	 q^{3  - 2m}, &\text{if } m >   |U| /2\\
\end{cases} \\
&\leq \begin{cases}
	 n^{- 2m + 3(1+\eps) -    3(1 + \eps) (1+1/ \eps) /2} , &\text{if }  m \leq |U| /2\\
	 n^{-2m + 3(1+\eps)  - 2\eps m}, &\text{if } m >  |U| /2\\
\end{cases} \\
&= o(n^{-2m}),
\end{align*}
where the last step follows from the assumption that $\eps < 1/4 < 1/2$ and since $2\eps m > 3(1+\eps)$ if  $m > |U| / 2 \geq  \lceil 3(1 +1 / \eps) \rceil /2$.

Note that we did not use that $|S| \leq s(\eps)$ in the proof above. The choice of the constant $s$ will become clear later in the proof. 
\medskip

\noindent \textbf{$\mathtt{hole}$:}
A configuration $(C_1, \ldots, C_m,  \mathtt{rule})$ of the type $\mathtt{hole}$ is defined as illustrated in the bottom left part of Figure \ref{fig:configurations}. The key characteristic are:
\begin{itemize}
	\item  $S$ encloses $U$, and
	\item $|U| \leq s(\eps)^2$. 
\end{itemize}
We bound the probability that all edges within $S \cup U$ are monochromatic. 

%To do this we apply Propositions~\ref{prop:GoodEdges} and \ref{prop:NewEdgePairs}  to certain sets of (new) edges within $S \cup U$. Recall that, by assumption, all edges between $S$ and $U$ are new.

We first consider the top-left piece in $U$ from every component $C_i$, $2\le i\le m$, i.e. the first piece from each $C_i$ we encounter while traversing $U$ row by row, from left to right. This implies that all these $m-1$ pieces have a new edge going up and  left. Let $W \subseteq U$ denote the set of theses pieces. % so that $|W|=m-1$, and all pieces in $W$ are a first piece of some component $C_i$.
Let ${\cal L}$ denote the set of edges between the pieces in $W$ and their left and upper neighbour (which are in $U$ or in $S$). Then $|{\cal L}| = 2(m-1)$  and Proposition~\ref{prop:GoodEdges} implies that the probability that all these edges are valid is $ q^{-|{\cal L}|}$.

Next, %let $S' \subset S$ be the puzzle pieces along the right and lower part of this enclosing border and consider the set 
let ${\cal L}'$ be the set of all the edges between $U$ and $S$ facing down or right. %between pieces $U \setminus W$ and $S'$. 
Note that $|{\cal L}'|= b/2$, where $b$ is the number of edges between $S$ and $U$ (i.e. the border of $U$). By fixing the colouring for all edges in ${\cal L}$, every edge in ${\cal L}'$ still has one half-edge of which the colour has not yet been revealed (again, $W$ does not contain any piece from $C_1$). Proposition~\ref{prop:NewEdgePairs} thus implies that the probability that all these edges are valid is at most $q^{-|{\cal L}'|/2}$. Thus, if $q = n^{1+\eps}$ we have
$$\Pr[X_{(C_1, \ldots, C_m, \mathtt{rule})}=1]  =  q^{-|{\cal L}|} q^{-|{\cal L}'|/2} = q^{2 - 2m - b/4}.$$
If $b \ge 8$ then we immediately get the desired bound
$$\Pr[X_{(C_1, \ldots, C_m, \mathtt{rule})}=1] = o(n^{-2m}).$$
Else, since the border of a connected set $U$ is always of even size, we have that $b \le 6$. Hence $U$ contains at most 2 puzzle pieces and the desired probability can be obtained by a simple case analysis.
\medskip

\noindent \textbf{$\mathtt{indentation}$:}
A configuration $(C_1, \ldots, C_m,  \mathtt{rule})$ of the type $\mathtt{indentation}$ is defined as illustrated in the bottom right part of Figure \ref{fig:configurations}. The key characteristics are:
\begin{itemize}
	\item $U$ is enclosed by $S$ from either two or three sides (one can imagine that $S$ together with one or two adjacent sides of the border of the grid enclose $U$),  % the right border of the grid encloses  another set of pieces $U$. 
	\item $|U| \leq s(\eps)^2$, and 
	\item each piece from $C_1, \ldots, C_m$ %originate from the `inside' of the original grid, but only from a zone around the border of $G$ of constant depth. 
	is at distance at most $s(\eps)^2$ from the grid border in $G$. 
\end{itemize}
Due to the third restriction, the number of configurations $(C_1, \ldots, C_m, \mathtt{rule})$ of the type $\mathtt{indentation}$ is $O(n^{m})$, since we only have linearly many pieces available to `build' the connected sets $C_1, \ldots, C_m$. It thus suffices to show that $\Pr[X_{(C_1, \ldots, C_m, \mathtt{rule})}=1] = o(n^{-m})$. 

%To do this we apply Propositions~\ref{prop:GoodEdges} and \ref{prop:NewEdgePairs}  to certain sets of (new) edges within $S \cup U$. Recall that, by assumption, all edges between $S$ and $U$ are new.

Similarly as in the previous case, we first consider the top-left piece in $U$ from every component $C_i$, $2\le i\le m$. Let $W$ denote the set of such pieces and % where `first' means in the order of traversing $U$ row by row. Thus, all these $m-1$ pieces have a new edge going up and going left. So let $W \subseteq U$ be a set of pieces so that $|W|=m-1$, and all pieces in $W$ are a first piece of some component $C_i$.
let ${\cal L}$ denote the set of edges between the pieces in $W$ and their left and upper neighbour (which are in $U$ or in $S$). Then $|{\cal L}| = 2(m-1)$  and, as before, Proposition~\ref{prop:GoodEdges} implies that the probability that all these edges are valid is $ q^{-|{\cal L}|}$. Hence, since at least one new  edge is present in $S \cup U$,  
$$\Pr[X_{(C_1, \ldots, C_m, \mathtt{rule})}=1]  \leq \min \{ {1/q, q^{-2(m-1)}} \} = o(n^{-m}).$$
\medskip

\noindent \textbf{$\mathtt{subsquare}$:}
The key characteristics of the configuration $(C_1, \ldots, C_m,  \mathtt{rule})$ of the type $\mathtt{subsquare}$ are:
\begin{itemize}
	\item $S$ is an empty set,
	\item $U$ is a $K(\varepsilon) \times K(\varepsilon)$ square where $K(\varepsilon) := 4s(\varepsilon)^2 / \varepsilon$, and
	\item stable sets of $U$ are of size at most $s(\eps)^2$. 
\end{itemize}
In the last property we restrict the notion of a stable set only to the rearrangement $U$, i.e. we say that a connected subset $U' \subseteq U$ is stable if and only if all edges in $E_U(U')$ are original edges and all edges in $E_U(U', U \setminus U')$ are new edges. Additionally, we denote by $U = \bigcup_{i=1}^{z} T_i$ the decomposition of $U$ into these restricted stable sets.
%Note that along the boundary of $U$ we might `loose' the overhanging pieces of the stable sets covering $U$, hence, by the third restriction, we know that $U = \bigcup_{i=1}^{z}$ decomposes into stable sets and subsets of stable sets of size at most $s(\eps)^2$. Since we will not consider the edges going outside of $U$, we will think of all the $T_i$'s as stable sets, slightly abusing notation.
Note that each $C_i$ is the union of some of the $T_j$'s and that each stable set $T_j$ belongs to exactly one $C_i$. %Moreover, we assume that $U$ is a $K \times K$ square, where $K$ is sufficiently large constant to be determined later. 

To show that $\Pr[X_{(C_1, \ldots, C_m, \mathtt{rule})}=1] = o(n^{-2m})$, we consider the following two cases:

\noindent \textbf{Case 1:}
$m >  2K(\varepsilon)(1 + 1 / \eps)$:\\
Consider the top-left piece in $U$ from every component $C_i$, $1 \leq i \leq m$. %where `first' means in the order of traversing the square row by row.
Note that all such pieces which do not touch the upper or the left border of the square have two new edges, going up and going left. Let $W \subseteq U$ be the set of these pieces and note that $|W|\geq m -2K(\varepsilon)$. %, all pieces in $W$ have a new edge going up and going left, and all pieces in $W$ are a first piece of some component $C_i$. 
Let ${\cal L}$ denote the set of edges between the pieces in $W$ and their left and upper neighbour. Then $|{\cal L}| \geq 2(m -2K(\varepsilon))$  and Proposition~\ref{prop:GoodEdges} implies that the probability that all these edges are valid is $ q^{-|{\cal L}|}$.
Thus, we have that
$$
	\Pr[X_{(C_1, \ldots, C_m, \mathtt{rule})}=1]  \leq q^{-2(m-2K(\varepsilon))}.
$$
Note that $m >2K(\varepsilon)(1 + 1 / \eps)$ implies $2m < 2(1 + \eps)(m-2K(\varepsilon))$. Hence, since $q = n^{1+ \eps}$, it follows that
$$\Pr[X_{(C_1, \ldots, C_m, \mathtt{rule})}=1] = o(n^{-2m}).$$

\noindent \textbf{Case 2:} $m \leq 2K(\varepsilon)(1 + 1 / \eps)$:\\
In this case, instead of considering the sets $C_i$, we will consider the sets~$T_i$. Consider the top-left piece in $U$ from every stable set $T_i$, $1 \leq i \leq z$. % where `first' means in the order of traversing the square row by row. 
Note that all such pieces which do not touch the upper or the left border of $U$ have a new edge going up and going left. So let $W \subseteq U$ be a set of these pieces and note that $|W|\geq z -2K(\varepsilon)$. %, all pieces in $W$ have a new edge going up and going left, and all pieces in $W$ are a first piece of some stable set $T_i$. 
Let ${\cal L}$ denote the set of edges between the pieces in $W$ and their left and upper neighbour. Then $|{\cal L}| \geq 2(z -2K(\varepsilon))$  and  Proposition~\ref{prop:NewEdgePairs} implies that the probability that all these edges are valid is $ q^{-|{\cal L}|/2}$.
Thus, we have that
$$
	\Pr[X_{(C_1, \ldots, C_m, \mathtt{rule})}=1]  \leq q^{-(z-2K(\varepsilon))}.
$$
Since $|T_i| \leq s(\eps)^2$ we have $z \geq K(\varepsilon)^2 / s(\varepsilon)^2$ and it is easy to check that the choice of $K$ implies $z-2K(\varepsilon) \geq  K(\varepsilon)^2/ s(\varepsilon)^2 - 2K(\varepsilon)  > 4K(\varepsilon)/ \eps$, with room to spare. Hence, since $m \leq 2K(\varepsilon)(1 + 1 / \eps)$, it holds that $(1 + \eps)(z-2K(\varepsilon))  > 2m$. Therefore, if $q = n^{1+\eps}$, we have that
$$\Pr[X_{(C_1, \ldots, C_m, \mathtt{rule})}=1]  \leq n^{-(1 + \eps)(z-2K)} = o(n^{-2m}).$$

Finally, note that $K(\varepsilon)$ is sufficiently large such that $|\bigcup_{i \in [m]} C_i| \le K(\varepsilon)^2$ for all five types of rules. With all these preliminaries at hand we are now ready to complete the proof of Theorem \ref{thm:Jigsaw}.

\begin{proof}[Proof of Theorem \ref{thm:Jigsaw}, $q \ge n^{1 + \varepsilon}$] Recall that a configuration is valid if all new edges defined by the given rule are monochromatic. We showed that the probability of a configuration being valid is $o(n^{-2m})$, respectively $o(n^{-m})$ for configurations of the type $\mathtt{indentation}$. Since there are $O(n^{m})$ possible configurations of the type $\mathtt{indentation}$ and $O(n^{2m})$ configurations of the other types, and because $m \le |S \cup U| \le K(\varepsilon)^2$, by Markov's inequality and a union-bound we obtain that w.h.p. none of the above mentioned configurations are valid. In other words, none of these configurations can appear in a valid reconstruction. Using this assumption, we show that a valid reconstruction is identical to $G$ (up to rotation of the whole grid). 

For the rest of the proof, let $H = G(\phi, \Pi)$ be a valid reconstruction of $G$. %By showing that certain properties will w.h.p. be satisfied in a valid reconstruction
We first show that a valid reconstruction contains no stable set $T$ of size at least $s(\varepsilon)^2$.
%and then use this to argue that $H^{\Pi}_{\phi}$ is identical to the original puzzle $G$ (up to rotation of the whole grid).

Let us assume, towards a contradiction, that there exists a stable set $T$ of size at least $s(\varepsilon)^2$ and let $R$ be the bounding rectangle of $T$. Let us first consider the case where at least one side of $R$ does not coincide with the border of the grid $H$. Note that this implies that there has to exist a side of $R$ which does not coincide with the border of the grid {\em and} has length at least $s(\varepsilon)$. Consider the corresponding part of the border of $T$, as defined in Lemma~\ref{lemma:observations}, part $(ii)$. Then this part of the boundary of $T$ also has size at least $s(\varepsilon)$. Let us denote the pieces along this border part by $S$. Note that $S$ cannot have a straight line segment of size $\ell(\varepsilon)$ since this would form a valid $\mathtt{straightline}$ configuration. Since the number of convex and concave corners along $S$ is the same (Lemma \ref{lemma:observations}, $(ii)$) and for every $\ell(\varepsilon)$ consecutive pieces along $S$ we have at least one corner, we conclude that there exists a connected subset $S' \subseteq S$ of size at most $s(\varepsilon) = 4\ell(\varepsilon)^2$ which contains $2 \ell(\varepsilon)$ convex corners. However, this forms a valid $\mathtt{convexcorners}$ configuration, which is again a contradiction. 

Therefore, we conclude that if such a set $T$ exists, then the bounding rectangle $R$ has to touch all four borders of the grid $H$, implying that $R$ is identical to the grid $H$. Note that in this case $T$ touches all four borders of the grid $H$ as well (see Figure \ref{fig:GiantSet} below).

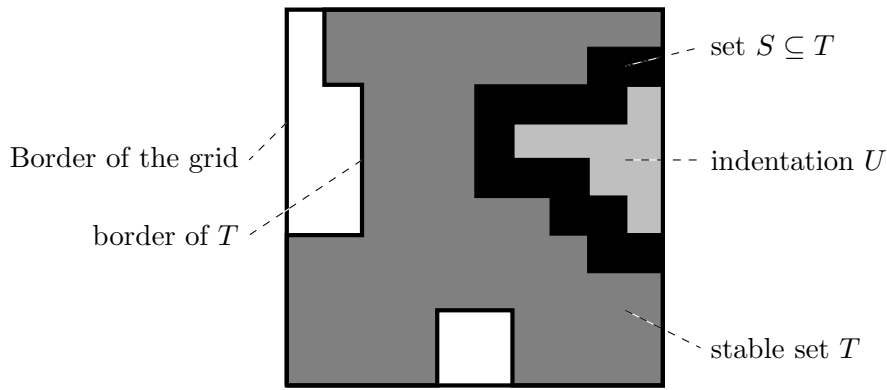
\begin{figure}[H]
\begin{center}
\begin{tikzpicture}

%connected component in gray%
\path [fill=gray] (0,0) -- (2,0) -- (2,1) -- (3,1) -- 
 	(3,0) -- (5,0) -- (5,2) -- (4.5, 2) -- (4.5,2.5) -- (4,2.5)
 	--	(4,3) -- (3,3) -- (3,3.5) -- (4.5, 3.5) -- (4.5,4) -- (5,4) -- (5,5) --
 	 (0.5, 5)-- (0.5,4) -- (1,4) --  (1,2) -- (0, 2) -- (0,0);

\draw  [white](4.5,1) -- (5.5,0.5); 	 
\draw  [dashed](4.5,1) -- (5.5,0.5);
\draw (5.5,0.5) node[right]  {stable set $T$};

%indentation U%
\path [fill=lightgray] (5,2) -- (4.5, 2) -- (4.5,2.5) -- (4,2.5)
 	--	(4,3) -- (3,3) -- (3,3.5) -- (4.5, 3.5) -- (4.5,4) -- (5,4) -- (5,2);

\draw  [white](4.5,3) -- (5.5,3);					
\draw  [dashed](4.5,3) -- (5.5,3);
\draw (5.5,3) node[right]  {indentation $U$};

%set S in green%

\draw [fill] (5,2) -- (4.5, 2) -- (4.5,2.5) -- (4,2.5)
 	--	(4,3) -- (3,3) -- (3,3.5) -- (4.5, 3.5) -- (4.5,4) -- (5,4)	
 	-- (5,4.5) -- (4, 4.5) -- (4, 4) -- (2.5, 4) -- (2.5,2.5) -- (3.5, 2.5)
 	-- (3.5,2) -- (4,2) -- (4, 1.5) -- (5,1.5) -- (5,2);	
 
\draw [white](4.5, 4.25) -- (5.5, 4.5);  
\draw [dashed](4.5, 4.25) -- (5.5, 4.5); 
\draw (5.5,4.5) node[right]  {set $S \subseteq T$};
 
%Grid in blue%
  \draw [ultra thick] (0,0) rectangle (5,5);

%Border of connected component in red%
 \draw [ultra thick] (0,0) -- (2,0) -- (2,1) -- (3,1) -- 
 	(3,0) -- (5,0) -- (5,2) -- (4.5, 2) -- (4.5,2.5) -- (4,2.5)
 	--	(4,3) -- (3,3) -- (3,3.5) -- (4.5, 3.5) -- (4.5,4) -- (5,4) -- (5,5) --
 	 (0.5, 5)-- (0.5,4) -- (1,4) --  (1,2) -- (0, 2) -- (0,0);
 	 
%labeling of important things%
\draw [white](0,3.5) -- (-0.5,3);
\draw [dashed](0,3.5) -- (-0.5,3);
\draw (-0.5,3) node[left]  {Border of the grid};

\draw  [white](1,3) -- (-0.5,2);
\draw  [dashed](1,3) -- (-0.5,2);
\draw (-0.5,2) node[left]  {border of $T$};

\end{tikzpicture}
\end{center}
\caption{Example of a `giant' stable set $T$}
\label{fig:GiantSet}
\end{figure}

Consider an indentation $U$ in $T$ and let $S \subseteq T$ be the set of pieces along the border of $U$ (see Figure \ref{fig:GiantSet}). Note that $U$ has four more concave than convex corners (Lemma \ref{lemma:observations}, part $(i)$). Moreover, there are at most three corners of $U$ which touch the boundary of the grid and thus there are more concave corners lying `inside' of the grid than the total number of convex corners of $U$. Since every such concave corner `seen' from the set $U$ is a convex corner `seen' from the set $S$ we conclude that $S$ contains more convex than concave corners along its border with $U$. A similar argument as in the previous case then shows that the set $S$ is of size at most $s(\eps)$ as otherwise there exists either a valid $\mathtt{straightline}$ or a valid $\mathtt{convexcorners}$ configuration. Therefore, every indentation $U$ is of size at most $s(\eps)^2$.
%, due to $(P1)$, w.h.p. these indentations are only of constant depth.
Similarly, the set $S$ of pieces around a hole is of size at most $s(\eps)$, and thus every hole of $T$ has to be of size at most $s(\varepsilon)^2$. However, since there are no valid $\mathtt{hole}$ configurations, we conclude that $T$ has no hole at all. %On the other hand, $T$ contains no holes (or else there exists a valid $\mathtt{hole}$ configuration), 
Therefore, the pieces of an indentation (like $U$ in Figure \ref{fig:GiantSet}) cannot originate from the `inside' of the grid, but only from a zone around the border of the grid of depth at most $s(\varepsilon)^2$. Finally, the existence of such indentations is excluded since there are no valid $\mathtt{indentation}$ configurations. Therefore, since a stable set $T$ touches all four borders of the grid and has no holes and indentations, we conclude that $T$ is the whole grid $G$.

To summarise, we showed that a valid reconstruction which is not identical to $G$ (up to rotation of the whole grid) contains only stable sets of size smaller than $s(\varepsilon)^2$. However, since there are no valid $\mathtt{subsquare}$ configurations we in turn conclude that no $K(\varepsilon) \times K(\varepsilon)$  subsquare in a valid reconstruction contains a new edge (as otherwise it would give a valid $\mathtt{subsquare}$ configuration). By shifting a subsquare of size $K(\varepsilon) \times K(\varepsilon)$ through the whole puzzle we thus deduce that $E_H$ contains no new edge at all. That is, $H$ is identical to the original puzzle $G$ up to rotation of the whole grid.
\end{proof}

To finish, let us make two observations about possible extensions of the proof. First, we see no reason why our arguments should not generalize to higher dimensions. Second, note that the above proof for the upper bound can be slightly adapted so that everything holds for suitable (though very slowly decreasing) $\eps = o(1)$ as well. It would be interesting to determine this more precisely. In particular, it is tempting to conjecture that actually $q = n\cdot \omega(n)$ suffices for  unique reconstruction, where $\omega(n)$ is some slowly growing function in $n$. %We leave this to future work.

\bibliographystyle{plain}
\bibliography{JigsawPuzzle}

\end{document}